\pdfoutput=1

\documentclass{amsart}

\usepackage[utf8]{inputenc}

\usepackage{hypercrossref}

\usepackage{ownstyle}

\begin{document}

\makeatletter
\newcommand{\sectionNotes}{\phantomsection\section*{Notes}\addcontentsline{toc}{section}{Notes}\markright{\textsc{\@chapapp{} \thechapter{} Notes}}}
\newcommand{\sectionExercises}[1]{\phantomsection\section*{Exercises}\addcontentsline{toc}{section}{Exercises}\markright{\textsc{\@chapapp{} \thechapter{} Exercises}}}
\makeatother

\newcommand{\jdeq}{\equiv}      
\let\judgeq\jdeq
\newcommand{\defeq}{\vcentcolon\equiv}  

\newcommand{\define}[1]{\textbf{#1}}


\newcommand{\Vect}{\ensuremath{\mathsf{Vec}}}
\newcommand{\Fin}{\ensuremath{\mathsf{Fin}}}
\newcommand{\fmax}{\ensuremath{\mathsf{fmax}}}
\newcommand{\seq}[1]{\langle #1\rangle}

\def\prdsym{\textstyle\prod}
\makeatletter
\def\prd#1{\@ifnextchar\bgroup{\prd@parens{#1}}{%
    \@ifnextchar\sm{\prd@parens{#1}\@eatsm}{%
    \@ifnextchar\prd{\prd@parens{#1}\@eatprd}{%
    \@ifnextchar\;{\prd@parens{#1}\@eatsemicolonspace}{%
    \@ifnextchar\\{\prd@parens{#1}\@eatlinebreak}{%
    \@ifnextchar\narrowbreak{\prd@parens{#1}\@eatnarrowbreak}{%
      \prd@noparens{#1}}}}}}}}
\def\prd@parens#1{\@ifnextchar\bgroup%
  {\mathchoice{\@dprd{#1}}{\@tprd{#1}}{\@tprd{#1}}{\@tprd{#1}}\prd@parens}%
  {\@ifnextchar\sm%
    {\mathchoice{\@dprd{#1}}{\@tprd{#1}}{\@tprd{#1}}{\@tprd{#1}}\@eatsm}%
    {\mathchoice{\@dprd{#1}}{\@tprd{#1}}{\@tprd{#1}}{\@tprd{#1}}}}}
\def\@eatsm\sm{\sm@parens}
\def\prd@noparens#1{\mathchoice{\@dprd@noparens{#1}}{\@tprd{#1}}{\@tprd{#1}}{\@tprd{#1}}}
\def\lprd#1{\@ifnextchar\bgroup{\@lprd{#1}\lprd}{\@@lprd{#1}}}
\def\@lprd#1{\mathchoice{{\textstyle\prod}}{\prod}{\prod}{\prod}({\textstyle #1})\;}
\def\@@lprd#1{\mathchoice{{\textstyle\prod}}{\prod}{\prod}{\prod}({\textstyle #1}),\ }
\def\tprd#1{\@tprd{#1}\@ifnextchar\bgroup{\tprd}{}}
\def\@tprd#1{\mathchoice{{\textstyle\prod_{(#1)}}}{\prod_{(#1)}}{\prod_{(#1)}}{\prod_{(#1)}}}
\def\dprd#1{\@dprd{#1}\@ifnextchar\bgroup{\dprd}{}}
\def\@dprd#1{\prod_{(#1)}\,}
\def\@dprd@noparens#1{\prod_{#1}\,}

\def\@eatnarrowbreak\narrowbreak{%
  \@ifnextchar\prd{\narrowbreak\@eatprd}{%
    \@ifnextchar\sm{\narrowbreak\@eatsm}{%
      \narrowbreak}}}
\def\@eatlinebreak\\{%
  \@ifnextchar\prd{\\\@eatprd}{%
    \@ifnextchar\sm{\\\@eatsm}{%
      \\}}}
\def\@eatsemicolonspace\;{%
  \@ifnextchar\prd{\;\@eatprd}{%
    \@ifnextchar\sm{\;\@eatsm}{%
      \;}}}

\def\lam#1{{\lambda}\@lamarg#1:\@endlamarg\@ifnextchar\bgroup{.\,\lam}{.\,}}
\def\@lamarg#1:#2\@endlamarg{\if\relax\detokenize{#2}\relax #1\else\@lamvar{\@lameatcolon#2},#1\@endlamvar\fi}
\def\@lamvar#1,#2\@endlamvar{(#2\,{:}\,#1)}
\def\@lameatcolon#1:{#1}
\let\lamt\lam
\def\lamu#1{{\lambda}\@lamuarg#1:\@endlamuarg\@ifnextchar\bgroup{.\,\lamu}{.\,}}
\def\@lamuarg#1:#2\@endlamuarg{#1}

\def\fall#1{\forall (#1)\@ifnextchar\bgroup{.\,\fall}{.\,}}

\def\exis#1{\exists (#1)\@ifnextchar\bgroup{.\,\exis}{.\,}}

\def\smsym{\textstyle\sum}
\def\sm#1{\@ifnextchar\bgroup{\sm@parens{#1}}{%
    \@ifnextchar\prd{\sm@parens{#1}\@eatprd}{%
    \@ifnextchar\sm{\sm@parens{#1}\@eatsm}{%
    \@ifnextchar\;{\sm@parens{#1}\@eatsemicolonspace}{%
    \@ifnextchar\\{\sm@parens{#1}\@eatlinebreak}{%
    \@ifnextchar\narrowbreak{\sm@parens{#1}\@eatnarrowbreak}{%
        \sm@noparens{#1}}}}}}}}
\def\sm@parens#1{\@ifnextchar\bgroup%
  {\mathchoice{\@dsm{#1}}{\@tsm{#1}}{\@tsm{#1}}{\@tsm{#1}}\sm@parens}%
  {\@ifnextchar\prd%
    {\mathchoice{\@dsm{#1}}{\@tsm{#1}}{\@tsm{#1}}{\@tsm{#1}}\@eatprd}%
    {\mathchoice{\@dsm{#1}}{\@tsm{#1}}{\@tsm{#1}}{\@tsm{#1}}}}}
\def\@eatprd\prd{\prd@parens}
\def\sm@noparens#1{\mathchoice{\@dsm@noparens{#1}}{\@tsm{#1}}{\@tsm{#1}}{\@tsm{#1}}}
\def\lsm#1{\@ifnextchar\bgroup{\@lsm{#1}\lsm}{\@@lsm{#1}}}
\def\@lsm#1{\mathchoice{{\textstyle\sum}}{\sum}{\sum}{\sum}({\textstyle #1})\;}
\def\@@lsm#1{\mathchoice{{\textstyle\sum}}{\sum}{\sum}{\sum}({\textstyle #1}),\ }
\def\tsm#1{\@tsm{#1}\@ifnextchar\bgroup{\tsm}{}}
\def\@tsm#1{\mathchoice{{\textstyle\sum_{(#1)}}}{\sum_{(#1)}}{\sum_{(#1)}}{\sum_{(#1)}}}
\def\dsm#1{\@dsm{#1}\@ifnextchar\bgroup{\dsm}{}}
\def\@dsm#1{\sum_{(#1)}\,}
\def\@dsm@noparens#1{\sum_{#1}\,}

\def\wtypesym{{\mathsf{W}}}
\def\wtype#1{\@ifnextchar\bgroup%
  {\mathchoice{\@twtype{#1}}{\@twtype{#1}}{\@twtype{#1}}{\@twtype{#1}}\wtype}%
  {\mathchoice{\@twtype{#1}}{\@twtype{#1}}{\@twtype{#1}}{\@twtype{#1}}}}
\def\lwtype#1{\@ifnextchar\bgroup{\@lwtype{#1}\lwtype}{\@@lwtype{#1}}}
\def\@lwtype#1{\mathchoice{{\textstyle\mathsf{W}}}{\mathsf{W}}{\mathsf{W}}{\mathsf{W}}({\textstyle #1})\;}
\def\@@lwtype#1{\mathchoice{{\textstyle\mathsf{W}}}{\mathsf{W}}{\mathsf{W}}{\mathsf{W}}({\textstyle #1}),\ }
\def\twtype#1{\@twtype{#1}\@ifnextchar\bgroup{\twtype}{}}
\def\@twtype#1{\mathchoice{{\textstyle\mathsf{W}_{(#1)}}}{\mathsf{W}_{(#1)}}{\mathsf{W}_{(#1)}}{\mathsf{W}_{(#1)}}}
\def\dwtype#1{\@dwtype{#1}\@ifnextchar\bgroup{\dwtype}{}}
\def\@dwtype#1{\mathsf{W}_{(#1)}\,}

\newcommand{\suppsym}{{\mathsf{sup}}}
\newcommand{\supp}{\ensuremath\suppsym\xspace}

\def\wtypeh#1{\@ifnextchar\bgroup%
  {\mathchoice{\@lwtypeh{#1}}{\@twtypeh{#1}}{\@twtypeh{#1}}{\@twtypeh{#1}}\wtypeh}%
  {\mathchoice{\@@lwtypeh{#1}}{\@twtypeh{#1}}{\@twtypeh{#1}}{\@twtypeh{#1}}}}
\def\lwtypeh#1{\@ifnextchar\bgroup{\@lwtypeh{#1}\lwtypeh}{\@@lwtypeh{#1}}}
\def\@lwtypeh#1{\mathchoice{{\textstyle\mathsf{W}^h}}{\mathsf{W}^h}{\mathsf{W}^h}{\mathsf{W}^h}({\textstyle #1})\;}
\def\@@lwtypeh#1{\mathchoice{{\textstyle\mathsf{W}^h}}{\mathsf{W}^h}{\mathsf{W}^h}{\mathsf{W}^h}({\textstyle #1}),\ }
\def\twtypeh#1{\@twtypeh{#1}\@ifnextchar\bgroup{\twtypeh}{}}
\def\@twtypeh#1{\mathchoice{{\textstyle\mathsf{W}^h_{(#1)}}}{\mathsf{W}^h_{(#1)}}{\mathsf{W}^h_{(#1)}}{\mathsf{W}^h_{(#1)}}}
\def\dwtypeh#1{\@dwtypeh{#1}\@ifnextchar\bgroup{\dwtypeh}{}}
\def\@dwtypeh#1{\mathsf{W}^h_{(#1)}\,}

\makeatother

\let\setof\Set    
\newcommand{\pair}{\ensuremath{\mathsf{pair}}\xspace}
\newcommand{\tup}[2]{(#1,#2)}
\newcommand{\proj}[1]{\ensuremath{\mathsf{pr}_{#1}}\xspace}
\newcommand{\fst}{\ensuremath{\proj1}\xspace}
\newcommand{\snd}{\ensuremath{\proj2}\xspace}
\newcommand{\ac}{\ensuremath{\mathsf{ac}}\xspace} 
\newcommand{\un}{\ensuremath{\mathsf{upun}}\xspace} 

\newcommand{\rec}[1]{\mathsf{rec}_{#1}}
\newcommand{\ind}[1]{\mathsf{ind}_{#1}}
\newcommand{\indid}[1]{\ind{=_{#1}}} 
\newcommand{\indidb}[1]{\ind{=_{#1}}'} 

\newcommand{\uppt}{\ensuremath{\mathsf{uppt}}\xspace}

\newcommand{\pairpath}{\ensuremath{\mathsf{pair}^{\mathord{=}}}\xspace}
\newcommand{\projpath}[1]{\ensuremath{\apfunc{\proj{#1}}}\xspace}

\newcommand{\pairr}[1]{{\mathopen{}(#1)\mathclose{}}}
\newcommand{\Pairr}[1]{{\mathopen{}\left(#1\right)\mathclose{}}}

\newcommand{\im}{\ensuremath{\mathsf{im}}} 

\newcommand{\leftwhisker}{\mathbin{{\ct}_{\mathsf{l}}}}  
\newcommand{\rightwhisker}{\mathbin{{\ct}_{\mathsf{r}}}} 
\newcommand{\hct}{\star}

\newcommand{\modal}{\ensuremath{\ocircle}}
\let\reflect\modal
\newcommand{\modaltype}{\ensuremath{\type_\modal}}
\newcommand{\mreturn}{\ensuremath{\eta}}
\let\project\mreturn
\newcommand{\ext}{\mathsf{ext}}
\renewcommand{\P}{\ensuremath{\type_{P}}\xspace}


\newcommand{\idsym}{{=}}
\newcommand{\id}[3][]{\ensuremath{#2 =_{#1} #3}\xspace}
\newcommand{\idtype}[3][]{\ensuremath{\mathsf{Id}_{#1}(#2,#3)}\xspace}
\newcommand{\idtypevar}[1]{\ensuremath{\mathsf{Id}_{#1}}\xspace}
\newcommand{\defid}{\coloneqq}

\newcommand{\dpath}[4]{#3 =^{#1}_{#2} #4}


\newcommand{\refl}[1]{\ensuremath{\mathsf{refl}_{#1}}\xspace}

\newcommand{\ct}{%
  \mathchoice{\mathbin{\raisebox{0.5ex}{$\displaystyle\centerdot$}}}%
             {\mathbin{\raisebox{0.5ex}{$\centerdot$}}}%
             {\mathbin{\raisebox{0.25ex}{$\scriptstyle\,\centerdot\,$}}}%
             {\mathbin{\raisebox{0.1ex}{$\scriptscriptstyle\,\centerdot\,$}}}
}

\newcommand{\opp}[1]{\mathord{{#1}^{-1}}}
\let\rev\opp

\newcommand{\trans}[2]{\ensuremath{{#1}_{*}\mathopen{}\left({#2}\right)\mathclose{}}\xspace}
\let\Trans\trans
\newcommand{\transf}[1]{\ensuremath{{#1}_{*}}\xspace} 
\newcommand{\transfib}[3]{\ensuremath{\mathsf{transport}^{#1}(#2,#3)\xspace}}
\newcommand{\Transfib}[3]{\ensuremath{\mathsf{transport}^{#1}\Big(#2,\, #3\Big)\xspace}}
\newcommand{\transfibf}[1]{\ensuremath{\mathsf{transport}^{#1}\xspace}}

\newcommand{\transtwo}[2]{\ensuremath{\mathsf{transport}^2\mathopen{}\left({#1},{#2}\right)\mathclose{}}\xspace}

\newcommand{\transconst}[3]{\ensuremath{\mathsf{transportconst}}^{#1}_{#2}(#3)\xspace}
\newcommand{\transconstf}{\ensuremath{\mathsf{transportconst}}\xspace}

\newcommand{\mapfunc}[1]{\ensuremath{\mathsf{ap}_{#1}}\xspace} 
\newcommand{\map}[2]{\ensuremath{{#1}\mathopen{}\left({#2}\right)\mathclose{}}\xspace}
\let\Ap\map
\newcommand{\mapdepfunc}[1]{\ensuremath{\mathsf{apd}_{#1}}\xspace} 
\newcommand{\mapdep}[2]{\ensuremath{\mapdepfunc{#1}\mathopen{}\left(#2\right)\mathclose{}}\xspace}
\let\apfunc\mapfunc
\let\ap\map
\let\apdfunc\mapdepfunc
\let\apd\mapdep

\newcommand{\aptwofunc}[1]{\ensuremath{\mathsf{ap}^2_{#1}}\xspace}
\newcommand{\aptwo}[2]{\ensuremath{\aptwofunc{#1}\mathopen{}\left({#2}\right)\mathclose{}}\xspace}
\newcommand{\apdtwofunc}[1]{\ensuremath{\mathsf{apd}^2_{#1}}\xspace}
\newcommand{\apdtwo}[2]{\ensuremath{\apdtwofunc{#1}\mathopen{}\left(#2\right)\mathclose{}}\xspace}

\newcommand{\idfunc}[1][]{\ensuremath{\mathsf{id}_{#1}}\xspace}

\newcommand{\htpy}{\sim}

\newcommand{\bisim}{\sim}       
\newcommand{\eqr}{\sim}         

\newcommand{\eqv}[2]{\ensuremath{#1 \simeq #2}\xspace}
\newcommand{\eqvspaced}[2]{\ensuremath{#1 \;\simeq\; #2}\xspace}
\newcommand{\eqvsym}{\simeq}    
\newcommand{\texteqv}[2]{\ensuremath{\mathsf{Equiv}(#1,#2)}\xspace}
\newcommand{\isequiv}{\ensuremath{\mathsf{isequiv}}}
\newcommand{\qinv}{\ensuremath{\mathsf{qinv}}}
\newcommand{\ishae}{\ensuremath{\mathsf{ishae}}}
\newcommand{\linv}{\ensuremath{\mathsf{linv}}}
\newcommand{\rinv}{\ensuremath{\mathsf{rinv}}}
\newcommand{\biinv}{\ensuremath{\mathsf{biinv}}}
\newcommand{\lcoh}[3]{\mathsf{lcoh}_{#1}(#2,#3)}
\newcommand{\rcoh}[3]{\mathsf{rcoh}_{#1}(#2,#3)}
\newcommand{\hfib}[2]{{\mathsf{fib}}_{#1}(#2)}

\newcommand{\total}[1]{\ensuremath{\mathsf{total}(#1)}}

\newcommand{\UU}{\ensuremath{\mathcal{U}}\xspace}
\let\bbU\UU
\let\type\UU
\newcommand{\typele}[1]{\ensuremath{{#1}\text-\mathsf{Type}}\xspace}
\newcommand{\typeleU}[1]{\ensuremath{{#1}\text-\mathsf{Type}_\UU}\xspace}
\newcommand{\typelep}[1]{\ensuremath{{(#1)}\text-\mathsf{Type}}\xspace}
\newcommand{\typelepU}[1]{\ensuremath{{(#1)}\text-\mathsf{Type}_\UU}\xspace}
\let\ntype\typele
\let\ntypeU\typeleU
\let\ntypep\typelep
\let\ntypepU\typelepU
\renewcommand{\set}{\ensuremath{\mathsf{Set}}\xspace}
\newcommand{\setU}{\ensuremath{\mathsf{Set}_\UU}\xspace}
\newcommand{\prop}{\ensuremath{\mathsf{Prop}}\xspace}
\newcommand{\propU}{\ensuremath{\mathsf{Prop}_\UU}\xspace}
\newcommand{\pointed}[1]{\ensuremath{#1_\bullet}}

\newcommand{\card}{\ensuremath{\mathsf{Card}}\xspace}
\newcommand{\ord}{\ensuremath{\mathsf{Ord}}\xspace}
\newcommand{\ordsl}[2]{{#1}_{/#2}}

\newcommand{\ua}{\ensuremath{\mathsf{ua}}\xspace} 
\newcommand{\idtoeqv}{\ensuremath{\mathsf{idtoeqv}}\xspace}
\newcommand{\univalence}{\ensuremath{\mathsf{univalence}}\xspace} 

\newcommand{\iscontr}{\ensuremath{\mathsf{isContr}}}
\newcommand{\contr}{\ensuremath{\mathsf{contr}}} 
\newcommand{\isset}{\ensuremath{\mathsf{isSet}}}
\newcommand{\isprop}{\ensuremath{\mathsf{isProp}}}

\let\hfiber\hfib

\newcommand{\trunc}[2]{\mathopen{}\left\Vert #2\right\Vert_{#1}\mathclose{}}
\newcommand{\ttrunc}[2]{\bigl\Vert #2\bigr\Vert_{#1}}
\newcommand{\Trunc}[2]{\Bigl\Vert #2\Bigr\Vert_{#1}}
\newcommand{\truncf}[1]{\Vert \blank \Vert_{#1}}
\newcommand{\tproj}[3][]{\mathopen{}\left|#3\right|_{#2}^{#1}\mathclose{}}
\newcommand{\tprojf}[2][]{|\blank|_{#2}^{#1}}
\def\pizero{\trunc0}

\newcommand{\brck}[1]{\trunc{}{#1}}
\newcommand{\bbrck}[1]{\ttrunc{}{#1}}
\newcommand{\Brck}[1]{\Trunc{}{#1}}
\newcommand{\bproj}[1]{\tproj{}{#1}}
\newcommand{\bprojf}{\tprojf{}}

\newcommand{\Parens}[1]{\Bigl(#1\Bigr)}

\let\extendsmb\ext
\newcommand{\extend}[1]{\extendsmb(#1)}

%
\newcommand{\emptyt}{\ensuremath{\mathbf{0}}\xspace}

\newcommand{\unit}{\ensuremath{\mathbf{1}}\xspace}
\newcommand{\ttt}{\ensuremath{\star}\xspace}

\newcommand{\bool}{\ensuremath{\mathbf{2}}\xspace}
\newcommand{\btrue}{{1_{\bool}}}
\newcommand{\bfalse}{{0_{\bool}}}

\newcommand{\inlsym}{{\mathsf{inl}}}
\newcommand{\inrsym}{{\mathsf{inr}}}
\newcommand{\inl}{\ensuremath\inlsym\xspace}
\newcommand{\inr}{\ensuremath\inrsym\xspace}

\newcommand{\seg}{\ensuremath{\mathsf{seg}}\xspace}

\newcommand{\freegroup}[1]{F(#1)}
\newcommand{\freegroupx}[1]{F'(#1)} 

\newcommand{\glue}{\mathsf{glue}}

\newcommand{\colim}{\mathsf{colim}}
\newcommand{\inc}{\mathsf{inc}}
\newcommand{\cmp}{\mathsf{cmp}}

\newcommand{\Sn}{\mathbb{S}}
\newcommand{\base}{\ensuremath{\mathsf{base}}\xspace}
\newcommand{\lloop}{\ensuremath{\mathsf{loop}}\xspace}
\newcommand{\surf}{\ensuremath{\mathsf{surf}}\xspace}

\newcommand{\susp}{\Sigma}
\newcommand{\north}{\mathsf{N}}
\newcommand{\south}{\mathsf{S}}
\newcommand{\merid}{\mathsf{merid}}

\newcommand{\blank}{\mathord{\hspace{1pt}\text{--}\hspace{1pt}}}

\newcommand{\nameless}{\mathord{\hspace{1pt}\underline{\hspace{1ex}}\hspace{1pt}}}

\newcommand{\bbP}{\ensuremath{\mathbb{P}}\xspace}

\newcommand{\uset}{\ensuremath{\mathcal{S}et}\xspace}
\newcommand{\ucat}{\ensuremath{{\mathcal{C}at}}\xspace}
\newcommand{\urel}{\ensuremath{\mathcal{R}el}\xspace}
\newcommand{\uhilb}{\ensuremath{\mathcal{H}ilb}\xspace}
\newcommand{\utype}{\ensuremath{\mathcal{T}\!ype}\xspace}

\newbox\pbbox
\setbox\pbbox=\hbox{\xy \POS(65,0)\ar@{-} (0,0) \ar@{-} (65,65)\endxy}
\def\pb{\save[]+<3.5mm,-3.5mm>*{\copy\pbbox} \restore}

\newcommand{\inv}[1]{{#1}^{-1}}
\newcommand{\idtoiso}{\ensuremath{\mathsf{idtoiso}}\xspace}
\newcommand{\isotoid}{\ensuremath{\mathsf{isotoid}}\xspace}
\newcommand{\op}{^{\mathrm{op}}}
\newcommand{\y}{\ensuremath{\mathbf{y}}\xspace}
\newcommand{\dgr}[1]{{#1}^{\dagger}}
\newcommand{\unitaryiso}{\mathrel{\cong^\dagger}}
\newcommand{\cteqv}[2]{\ensuremath{#1 \simeq #2}\xspace}
\newcommand{\cteqvsym}{\simeq}     

\newcommand{\N}{\ensuremath{\mathbb{N}}\xspace}
\let\nat\N
\newcommand{\natp}{\ensuremath{\nat'}\xspace} 

\newcommand{\zerop}{\ensuremath{0'}\xspace}   
\newcommand{\suc}{\mathsf{succ}}
\newcommand{\sucp}{\ensuremath{\suc'}\xspace} 
\newcommand{\add}{\mathsf{add}}
\newcommand{\ack}{\mathsf{ack}}
\newcommand{\ite}{\mathsf{iter}}
\newcommand{\assoc}{\mathsf{assoc}}
\newcommand{\dbl}{\ensuremath{\mathsf{double}}}
\newcommand{\dblp}{\ensuremath{\dbl'}\xspace} 

\newcommand{\lst}[1]{\mathsf{List}(#1)}
\newcommand{\nil}{\mathsf{nil}}
\newcommand{\cons}{\mathsf{cons}}
\newcommand{\lost}[1]{\mathsf{Lost}(#1)}

\newcommand{\vect}[2]{\ensuremath{\mathsf{Vec}_{#1}(#2)}\xspace}

\newcommand{\Z}{\ensuremath{\mathbb{Z}}\xspace}
\newcommand{\Zsuc}{\mathsf{succ}}
\newcommand{\Zpred}{\mathsf{pred}}

\newcommand{\Q}{\ensuremath{\mathbb{Q}}\xspace}

\newcommand{\funext}{\mathsf{funext}}
\newcommand{\happly}{\mathsf{happly}}

\newcommand{\com}[3]{\mathsf{swap}_{#1,#2}(#3)}

\newcommand{\code}{\ensuremath{\mathsf{code}}\xspace}
\newcommand{\encode}{\ensuremath{\mathsf{encode}}\xspace}
\newcommand{\decode}{\ensuremath{\mathsf{decode}}\xspace}

\newcommand{\function}[4]{\left\{\begin{array}{rcl}#1 &
      \longrightarrow & #2 \\ #3 & \longmapsto & #4 \end{array}\right.}

\newcommand{\cone}[2]{\mathsf{cone}_{#1}(#2)}
\newcommand{\cocone}[2]{\mathsf{cocone}_{#1}(#2)}
\newcommand{\composecocone}[2]{#1\circ#2}
\newcommand{\composecone}[2]{#2\circ#1}
\newcommand{\Ddiag}{\mathscr{D}}

\newcommand{\Map}{\mathsf{Map}}

\newcommand{\interval}{\ensuremath{I}\xspace}
\newcommand{\izero}{\ensuremath{0_{\interval}}\xspace}
\newcommand{\ione}{\ensuremath{1_{\interval}}\xspace}

\newcommand{\epi}{\ensuremath{\twoheadrightarrow}}
\newcommand{\mono}{\ensuremath{\rightarrowtail}}

\newcommand{\bin}{\ensuremath{\mathrel{\widetilde{\in}}}}

\newcommand{\semigroupstrsym}{\ensuremath{\mathsf{SemigroupStr}}}
\newcommand{\semigroupstr}[1]{\ensuremath{\mathsf{SemigroupStr}}(#1)}
\newcommand{\semigroup}[0]{\ensuremath{\mathsf{Semigroup}}}

\newcommand{\emptyctx}{\ensuremath{\cdot}}
\newcommand{\production}{\vcentcolon\vcentcolon=}
\newcommand{\conv}{\downarrow}
\newcommand{\ctx}{\ensuremath{\mathsf{ctx}}}
\newcommand{\wfctx}[1]{#1\ \ctx}
\newcommand{\oftp}[3]{#1 \vdash #2 : #3}
\newcommand{\jdeqtp}[4]{#1 \vdash #2 \jdeq #3 : #4}
\newcommand{\judg}[2]{#1 \vdash #2}
\newcommand{\tmtp}[2]{#1 \mathord{:} #2}

\newcommand{\form}{\textsc{form}}
\newcommand{\intro}{\textsc{intro}}
\newcommand{\elim}{\textsc{elim}}
\newcommand{\uniq}{\textsc{uniq}}
\newcommand{\Weak}{\mathsf{Wkg}}
\newcommand{\Vble}{\mathsf{Vble}}
\newcommand{\Exch}{\mathsf{Exch}}
\newcommand{\Subst}{\mathsf{Subst}}

\newcommand{\cc}{\mathsf{c}}
\newcommand{\pp}{\mathsf{p}}
\newcommand{\cct}{\widetilde{\mathsf{c}}}
\newcommand{\ppt}{\widetilde{\mathsf{p}}}
\newcommand{\Wtil}{\ensuremath{\widetilde{W}}\xspace}

\newcommand{\istype}[1]{\mathsf{is}\mbox{-}{#1}\mbox{-}\mathsf{type}}
\newcommand{\nplusone}{\ensuremath{(n+1)}}
\newcommand{\nminusone}{\ensuremath{(n-1)}}
\newcommand{\fact}{\mathsf{fact}}

\newcommand{\kbar}{\overline{k}} 

\newcommand{\natw}{\ensuremath{\mathbf{N^w}}\xspace}
\newcommand{\zerow}{\ensuremath{0^\mathbf{w}}\xspace}
\newcommand{\sucw}{\ensuremath{\mathsf{succ}^{\mathbf{w}}}\xspace}
\newcommand{\nalg}{\nat\mathsf{Alg}}
\newcommand{\nhom}{\nat\mathsf{Hom}}
\newcommand{\ishinitw}{\mathsf{isHinit}_{\mathsf{W}}}
\newcommand{\ishinitn}{\mathsf{isHinit}_\nat}
\newcommand{\w}{\mathsf{W}}
\newcommand{\walg}{\w\mathsf{Alg}}
\newcommand{\whom}{\w\mathsf{Hom}}

\newcommand{\RC}{\ensuremath{\mathbb{R}_\mathsf{c}}\xspace} 
\newcommand{\RD}{\ensuremath{\mathbb{R}_\mathsf{d}}\xspace} 
\newcommand{\R}{\ensuremath{\mathbb{R}}\xspace}           
\newcommand{\barRD}{\ensuremath{\bar{\mathbb{R}}_\mathsf{d}}\xspace} 

\newcommand{\close}[1]{\sim_{#1}} 
\newcommand{\closesym}{\mathord\sim}
\newcommand{\rclim}{\mathsf{lim}} 
\newcommand{\rcrat}{\mathsf{rat}} 
\newcommand{\rceq}{\mathsf{eq}_{\RC}} 
\newcommand{\CAP}{\mathcal{C}}    
\newcommand{\Qp}{\Q_{+}}
\newcommand{\apart}{\mathrel{\#}}  
\newcommand{\dcut}{\mathsf{isCut}}  
\newcommand{\cover}{\triangleleft} 
\newcommand{\intfam}[3]{(#2, \lam{#1} #3)} 

\newcommand{\bsim}{\frown}
\newcommand{\bbsim}{\smile}

\newcommand{\hapx}{\diamondsuit\approx}
\newcommand{\hapname}{\diamondsuit}
\newcommand{\hapxb}{\heartsuit\approx}
\newcommand{\hapbname}{\heartsuit}
\newcommand{\tap}[1]{\bullet\approx_{#1}\triangle}
\newcommand{\tapname}{\triangle}
\newcommand{\tapb}[1]{\bullet\approx_{#1}\square}
\newcommand{\tapbname}{\square}

\newcommand{\NO}{\ensuremath{\mathsf{No}}\xspace}
\newcommand{\surr}[2]{\{\,#1\,\big|\,#2\,\}}
\newcommand{\LL}{\mathcal{L}}
\newcommand{\RR}{\mathcal{R}}
\newcommand{\noeq}{\mathsf{eq}_{\NO}} 

\newcommand{\ble}{\trianglelefteqslant}
\newcommand{\blt}{\vartriangleleft}
\newcommand{\bble}{\sqsubseteq}
\newcommand{\bblt}{\sqsubset}

\newcommand{\hle}{\diamondsuit\preceq}
\newcommand{\hlt}{\diamondsuit\prec}
\newcommand{\hlname}{\diamondsuit}
\newcommand{\hleb}{\heartsuit\preceq}
\newcommand{\hltb}{\heartsuit\prec}
\newcommand{\hlbname}{\heartsuit}
\newcommand{\tle}{\triangle\preceq}
\newcommand{\tlt}{\triangle\prec}
\newcommand{\tlname}{\triangle}
\newcommand{\tleb}{\square\preceq}
\newcommand{\tltb}{\square\prec}
\newcommand{\tlbname}{\square}

\newcommand{\vset}{\mathsf{set}}  
\def\cd{\tproj0}
\newcommand{\inj}{\ensuremath{\mathsf{inj}}} 
\newcommand{\acc}{\ensuremath{\mathsf{acc}}} 

\newcommand{\atMostOne}{\mathsf{atMostOne}}

\newcommand{\power}[1]{\mathcal{P}(#1)} 
\newcommand{\powerp}[1]{\mathcal{P}_+(#1)} 


\numberwithin{equation}{section}




\newcommand{\Coq}{\textsc{Coq}\xspace}
\newcommand{\Agda}{\textsc{Agda}\xspace}
\newcommand{\NuPRL}{\textsc{NuPRL}\xspace}
\newcommand{\agdabase}{\textsc{agda-base}\xspace}
\newcommand{\mtypes}{\textsc{m-types}\xspace}




\newcommand{\footstyle}[1]{{\hyperpage{#1}}n} 
\newcommand{\defstyle}[1]{\textbf{\hyperpage{#1}}}  

\newcommand{\indexdef}[1]{\index{#1|defstyle}}   
\newcommand{\indexfoot}[1]{\index{#1|footstyle}} 
\newcommand{\indexsee}[2]{\index{#1|see{#2}}}    


\newcommand{\ZF}{Zermelo--Fraenkel}
\newcommand{\CZF}{Constructive \ZF{} Set Theory}

\newcommand{\LEM}[1]{\ensuremath{\mathsf{LEM}_{#1}}\xspace}
\newcommand{\choice}[1]{\ensuremath{\mathsf{AC}_{#1}}\xspace}


\title{Non-wellfounded trees in Homotopy Type Theory}

\author{Benedikt Ahrens} 
\author{Paolo Capriotti}
\author{R\'egis Spadotti}

\thanks{The work of Benedikt Ahrens was partially supported by the CIMI (Centre International de Mathématiques et d'Informatique)
Excellence program ANR-11-LABX-0040-CIMI within the program ANR-11-IDEX-0002-02 during a postdoctoral fellowship.}

\begin{abstract}
 We prove a conjecture about the constructibility of \emph{coinductive types}---in the principled form of \emph{indexed $\M$-types}---in Homotopy Type Theory.
 The conjecture says that in the presence of \emph{inductive} types, coinductive types are derivable.
 Indeed, in this work, we construct coinductive types in a subsystem of Homotopy Type Theory; 
 this subsystem is given by Intensional Martin-Löf type theory with natural numbers and Voevodsky's Univalence Axiom.
 Our results are mechanized in the computer proof assistant \Agda.  
\end{abstract}
  
\maketitle

\section{Introduction}

Coinductive data types are used in functional programming to represent infinite data structures.
Examples include the ubiquitous data type of streams over a given base type, but also more sophisticated types; 
as an example we present an alternative definition of equivalence of types (\Cref{ex:equivalence}).

From a categorical perspective, coinductive types are characterized by a \emph{universal property},
which specifies the object with that property \emph{uniquely} in a suitable sense.
More precisely, a coinductive type is specified as the \emph{terminal coalgebra} of a suitable endofunctor.
In this category-theoretic viewpoint, coinductive types are dual to \emph{inductive} types, which are 
defined as initial algebras.

Inductive, resp.\ coinductive, types are usually considered in the principled form of the family of $\W$-types, resp.\ $\M$-types, 
parametrized by a type $A$ and a dependent type family $B$ over $A$, that is, a family of types $(B(a))_{a:A}$. 
Intuitively, the elements of the coinductive type $\M(A,B)$ are trees with nodes labeled by elements of $A$ 
such that a node labeled by $a:A$ has $B(a)$-many subtrees, given by a map $B(a) \to \M(A,B)$; see \Cref{fig:tree} for an example.
The \emph{inductive} type $\W(A,B)$  contains only trees where any path within that tree
eventually leads to a \emph{leaf}, that is, to a node $a:A$ such that $B(a)$ is empty.

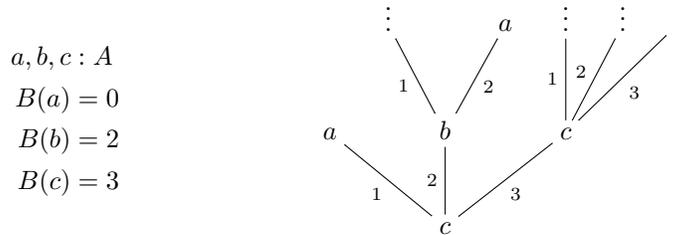
\begin{figure}[htb]
\begin{minipage}{0.3\textwidth}
 \begin{align*}
 a,b,c &: A\\
  B(a) &= 0\\
  B(b) &= 2\\
  B(c) &= 3
\end{align*}
\end{minipage}
\begin{minipage}{0.6\textwidth}
 \[
 \begin{xy}
  \xymatrix@C=1pc {
               & \vdots &  & a & \vdots & \vdots & \vdots    \\
            a  & & b \ar@{-}[lu]^{1} \ar@{-}[ru]_{2}& & c\ar@{-}[u]^{1}\ar@{-}[ru]^{2}\ar@{-}[rru]_{3} & &    \\
               & & c \ar@{-}[llu]^{1} \ar@{-}[u]^{2} \ar@{-}[rru]_{3} & & & &    
   }
 \end{xy}
\]
\end{minipage}
 \caption{Example of a tree (adapted from \cite{DBLP:journals/apal/BergM07})} \label{fig:tree}
\end{figure}

In this work, we study coinductive types in Homotopy Type Theory (HoTT). 
HoTT is an extension of intensional Martin-Löf type theory \cite{martin_lof}; we give a brief overview in \Cref{sec:hott}.

The universal properties defining inductive and coinductive types, respectively, 
can be expressed internally to intensional Martin-Löf type theory (and thus internally to HoTT).
Awodey, Gambino, and Sojakova \cite{DBLP:conf/lics/AwodeyGS12} use this facility when proving, within a subtheory $\mathcal{H}$ of HoTT, a logical equivalence between
\begin{enumerate}
 \item the existence of $\W$-types (a.k.a.\ the existence of a universal object) and
 \item the addition of a set of type-theoretic rules to their \enquote{base theory} $\mathcal{H}$. \label{enum:rules_ind_types}
\end{enumerate}
We might call the $\W$-types defined internally \enquote{internal $\W$-types}, and those specified via type-theoretic rules \enquote{external} ones.
In that sense, Awodey, Gambino, and Sojakova \cite{DBLP:conf/lics/AwodeyGS12} prove a logical equivalence between the existence of internal and external $\W$-types.

The universal property defining (internal) coinductive types in HoTT is dual to the one defining (internal) inductive types. 
One might hence assume that their existence is equivalent to a set of type-theoretic rules dual (in a suitable sense) to those given for external $\W$-types as in \cref{enum:rules_ind_types} above.
However, the rules for external $\W$-types cannot be dualized in a na\"ive way, due to some asymmetry of HoTT related to dependent types as maps into a \enquote{type of types} (a \emph{universe}),
see the discussion in \cite{hott-mailing-coind}.

In this work, we show instead that coinductive types in the form of $\M$-types can be derived from certain inductive types.
(More precisely, only one specific $\W$-type is needed: the type of natural numbers, which is readily specified as a $\W$-type \cite{DBLP:conf/lics/AwodeyGS12}.)

The result presented in this work is not surprising; indeed, the constructibility of coinductive types from inductive types 
has been shown in extensional type theory (see \Cref{sec:rel-work}) and
was conjectured to work in HoTT during a discussion on the HoTT mailing list \cite{hott-mailing-coind}.
In this work, we give a formal proof of the constructibility of a class of coinductive types from inductive types, with a proof of correctness of the construction.

The theorem we prove here is actually more general than described above: instead of plain $\M$-types as described above, we construct \emph{indexed} $\M$-types, which 
can be considered as a form of \enquote{(simply-)typed} trees, typed over a type of indices $I$. Plain $\M$-types then correspond to the mono-typed indexed $\M$-types,
that is, to those for which $I = 1$.
Since all the ideas are already contained in the case of plain $\M$-types, we describe the construction of those extensively, and only briefly state the definitions
and the main result for the indexed case. The formalisation in \Agda, however, is done for the more general, indexed, case.
An example illustrates the need for these more general \emph{indexed} $\M$-types.

\subsection{Related work}\label{sec:rel-work}
 Inductive types in the form of $\W$-types in HoTT have been studied by Awodey, Gambino, and Sojakova \cite{DBLP:conf/lics/AwodeyGS12}. The content of that work is described above.

 Van den Berg and De Marchi \cite{DBLP:journals/apal/BergM07} study the existence of \emph{plain} $\M$-types in models of \emph{extensional} type theory, that is, of type
 theory with a reflection rule identifying propositional and judgmental equality.
 They prove the derivability of $\M$-types from $\W$-types in such models, see Corollary 2.5 of the arXiv version of that article.
 A construction in extensional type theory of $\M$-types from $\W$-types is given by Abbott, Altenkirch, and Ghani \cite{DBLP:journals/tcs/AbbottAG05}.
 
 Martin-Löf type theory without identity reflection, but with the principle of \emph{Uniqueness of Identity Proofs} (Axiom K) can be identified with 
 the 0-truncated fragment of HoTT (modulo the assumption of univalence and HITs). For such a type theory,
 a construction of (indexed) $\M$-types from $\W$-types is described by Altenkirch et al.~\cite{indexed_containers},
 internalizing a standard result in 1-category theory \cite{DBLP:journals/tcs/Barr93}.
 The present work thus generalizes the construction described in \cite{indexed_containers} by extending it from the 0-truncated fragment to the whole of HoTT.
 More specifically, the main work in this generalization is to develop higher-categorical variants of the 1-categorical constructions used in 
 \cite{indexed_containers} that are compatible with the higher-categorical structure (the coherence data) of types.
\subsection{Synopsis}

The paper is organized as follows:
In \Cref{sec:hott} we present the type theory we are working in---a \enquote{subsystem} of HoTT as presented in \cite{hottbook}.
In \Cref{sec:m-in-hott} we define signatures for \emph{plain} $\M$-types and, via a universal property, the $\M$-type associated to a given signature.
In \Cref{sec:derive-m} we construct the $\M$-type of a given signature.
In \Cref{sec:indexed} we state the main result for the case of \emph{indexed} $\M$-types.
Finally, in \Cref{sec:formal} we give an overview of the formalisation of our result in the proof assistant \Agda.

\subsection*{Acknowledgments}
We are grateful to many people for helpful discussions about coinductive types, online as well as offline: 
 Thorsten Altenkirch, 
 Steve Awodey, 
 Mart\'{i}n Escard\'{o}, 
 Nicolai Kraus,
 Peter LeFanu Lumsdaine,
 Ralph Matthes, 
 Paige North,
 Mike Shulman, and
 Vladimir Voevodsky.
We thank Nicolai Kraus, Peter LeFanu Lumsdaine and Paige North for suggesting improvements and clarifications for this paper.

\section{The type theory under consideration}\label{sec:hott}

The present work takes place within a type theory that is a subsystem of the type theory presented in the HoTT book \cite{hottbook}.
The latter is often referred to as Homotopy Type Theory (HoTT); it is an extension of intensional Martin-L\"of type theory (IMLTT) \cite{martin_lof}.
The extension is given by two data: firstly, the \emph{Univalence Axiom}, introduced by Vladimir Voevodsky and proven consistent
with IMLTT in the simplicial set model \cite{simp_set_model}.
The second extension is given by \emph{Higher Inductive Types} (HITs), the precise theory of which is still subject to active research.
Preliminary results on HITs have been worked out by Sojakova \cite{DBLP:conf/popl/Sojakova15} and 
Lumsdaine and Shulman---see \cite[Chap.\ 6]{hottbook} for an introduction.
In the present work, we use the Univalence Axiom, but do not make use of HITs.

The syntax of HoTT is extensively described in a book \cite{hottbook}; 
we only give a brief summary of the type constructors used in the present work, thus fixing notation.
The fundamental objects are \emph{types}, which have \emph{elements} (\enquote{inhabitants}), written $a:A$.
Types can be dependent on terms, which we write as $x:A\entails B(x) : \type$. In the preceding judgment, we use a special type $\type$, 
the \enquote{universe} or \enquote{type of types}. In this work we  assume any type being an element of $\type$ 
in the sense of \enquote{typical ambiguity} \cite[Chap.\ 1.3]{hottbook}, without worrying about 
universe levels. The formalisation in \Agda ensures that everything works fine in that respect: 
as we will see later, the universe $\type$ is closed under the construction of $\M$-types.

We use the following type constructors: dependent products $\prd{x:A}B(x)$, with non-dependent variant written $A \to B$, 
dependent sums $\sm{x:A}B(x)$ with non-dependent variant written $A \times B$, the identity type $\id[A]{x}{y}$ and
the coproduct type $A + B$.
In particular, we assume the empty type $0$ and the singleton type $1$.
Furthermore, we assume the existence of a type of natural numbers, given as an inductive type according to the rules given in \cite[Chap.\ 1.9]{hottbook}.
Finally, we assume the univalence axiom for the universe $\type$ as presented in \cite[Chap.\ 2.10]{hottbook}.

Concerning terms, function application is denoted by parentheses as in $f(x)$ or, occasionally, simply by juxtaposition. 
We write dependent pairs as $(a,b)$ for $b : B(a)$. Projections are indicated by a subscript, that is,
for $x : \sm{a:A}B(a)$ we have $x_0 : A$ and $x_1 : B(x_0)$.
Indices are also used occasionally to specify earlier arguments of a function of several arguments; 
e.g., we write $B_i(a)$ instead of $B(i)(a)$.

We conclude this brief introduction by recalling two important internally definable properties of types:
  we call the type $X$ \fat{contractible}, if $X$ is inhabited by a unique element, that is, if the following type is inhabited:
  \[ \iscontr(X) := \sm{x : X} \prd{x' : X} x' = x \enspace . \]
  We call the type $Y$ \fat{a proposition} if for all $y,y':Y$, we have $y = y'$.
  Note that a type $X$ is contractible iff $X$ is a proposition and there is an element $x:X$ (see also \cite[Lemma 3.11.3]{hottbook}).

\section{Definition of $\M$-types via universal property}\label{sec:m-in-hott}

\emph{Coinductive} types represent \emph{potentially infinite} data structures, 
such as streams or infinite lists.
As such, they have to be contrasted to \emph{inductive} datatypes, which represent structures that are \emph{necessarily finite},
such as unary natural numbers or finite lists.

\subsection{Signatures, a.k.a.\ containers}

In order to analyze inductive and coinductive types systematically, one usually fixes a notion of \enquote{signature}: 
a signature specifies, in an abstract way, the rules according to which the instances of a data structure are built.
In the following, we consider signatures to be given by \enquote{containers} \cite{indexed_containers}:
\begin{definition}
A \fat{container} (or \fat{signature}) is a pair $(A,B)$ of a type $A$ and a dependent type $x:A\entails B(x):\type$ over $A$.
\end{definition}
The container $(A,B)$ then determines a type of \enquote{trees} built as follows:
such a tree consists of a \emph{root node}, labeled by an element $a:A$, and a family of trees---\enquote{subtrees} of the original tree---indexed by the type $B(a)$.
A tree is \emph{well-founded} if it does not have an infinite chain of subtrees.

To the container $(A,B)$ one associates two types of trees built according to those rules: the type $\W(A,B)$ of well-founded trees, 
and the type $\M(A,B)$ of all trees, i.e., not necessarily well-founded.

The description of the inhabitants of $\W(A,B)$ and $\M(A,B)$ in terms of trees gives a suitable intuition; formally, those types are defined in terms of a
\emph{universal property}.
Indeed, $\M(A,B)$ will be defined as (the carrier of) a terminal object in a suitable sense.

%

\subsection{Coalgebras for a signature}

Any container $(A,B)$ specifies an endomorphism on types as follows:

\begin{definition}\label{def:poly_functor}
 Given a container $(A,B)$, define the \fat{polynomial functor} $P : \type \to \type$ associated to $(A,B)$ as
 \[P(X):=P_{A,B}(X):= \sm{a:A}(B(a) \to X) \enspace .\] 
 Given a map $f : X \to Y$, define $Pf: PX \to PY$ as the map
 \[ Pf(a,g) := (a, f \circ g) \enspace . \]
\end{definition}

Note that \Cref{def:poly_functor} does not really define a functor, and, more fundamentally, the universe $\type$ is not a (pre-)category
in the sense of \cite{rezk_completion}.
Instead, the appropriate notion for $P$ would be an $\infty$-(endo)functor on the $(\infty,1)$-category $\type$ \cite{james_cranch}.
However, we do not attempt to make any of these notions precise, and do not make use of any \enquote{functorial} properties of the defined maps.
Our use of the word \enquote{functor} merely indicates an analogy to the 1-categorical case.

To any signature $S=(A,B)$ we associate a type of 
\emph{coalgebras} $\Coalg_{S}$, and a family of types of morphisms between them:

\begin{definition}\label{def:coalgebra}
 Given a signature $S = (A,B)$ as in \Cref{def:poly_functor}, an \fat{$S$-coalgebra} is defined to be a pair $(C,\gamma)$ consisting of 
 a type $C:\type$ and a map $\gamma:C \to P_{S}C$.
 A map of coalgebras from $(C,\gamma)$ to $(D,\delta)$ is defined to be a pair $(f,p)$ of a map $f : C\to D$ and a path $p : \comp{f}{\delta} = \comp{\gamma}{P_Sf} $.
 Put differently, we set
  \[ \Coalg_S := \sm{C:\type} C\to PC \] and
  \[ \Coalg_S\Bigl( (C,\gamma),(D,\delta) \Bigr) := \sm{f : C\to D} \comp{f}{\delta} = \comp{\gamma}{Pf}  \enspace . \]
\end{definition}

\noindent
There is an obvious composition of coalgebra morphisms, and the identity map $C \to C$ is the carrier of a coalgebra endomorphism on $(\C,\gamma)$.
We also write $\Coalgmor{(C,\gamma)}{(D,\delta)}$ for the type of coalgebra morphisms from $(C,\gamma)$ to $(D,\delta)$.

\subsection{What is an $\M$-type?}

In this section we define (internal) $\M$-types in HoTT via a universal property.

\begin{definition}\label{def:m-type}
 Given a container $(A,B)$, \fat{the (internal) $\M$-type $\M_{A,B}$ associated to $(A,B)$} 
 is defined to be the pair $(M,\out : M \to P_{A,B}M)$ with the following universal 
 property: for any coalgebra $(C,\gamma) : \Coalg_{(A,B)}$ of $(A,B)$, the type of coalgebra morphisms $\Coalgmor{(C,\gamma)}{(M,\out)}$ from $(C,\gamma)$ to $(M,\out)$ is contractible.
\end{definition}

The use of the definite article in \Cref{def:m-type} is justified by the following lemma:
\begin{lemma}\label{lem:is-prop-final}
 The type 
  \[\Final_S:= \sm{(X,\rho):\Coalg_S}\prd{(C,\gamma):\Coalg_S}\iscontr\big(\Coalgmor{(C,\gamma)}{(X,\rho)}\big)\] 
 is a proposition.
\end{lemma}
\begin{proof}
 The proof that any two final coalgebras $(L, \out)$ and $(L', \out')$ have equivalent carriers is standard. 
 The Univalence Axiom then implies that the carriers are (pro\-po\-sitionally) equal, $L = L'$. 
 It then remains to show that the coalgebra structure $\out$, when transported along this identity, is equal to $\out'$. 
 We refer to the formalized proof for details. 
\end{proof}

That is, any inhabitant of $\Final_S$ is necessarily unique up to propositional equality. 
We refer to this inhabitant as the \enquote{final coalgebra}.
In \Cref{sec:derive-m} we construct the final coalgebra.

In the introduction, we use the adjectives \enquote{internal} and \enquote{external} to distinguish between types specified via
universal properties and type-theoretic rules, respectively. Since we do not consider rules for $\M$-types (that is, external $\M$-types) in this work,
we drop the adjective \enquote{internal} in what follows.

In the following example, we anticipate the result of the next section, namely the existence of a final coalgebra for any signature $(A,B)$:
\begin{example}\label{example:streams}
The coinductive type $\stream(A_0)$ of streams over a base type $A_0$ is given
by $\M(A,B)$ with $A = A_0$ and $B(a):= 1$ for any $a:A_0$.  The corresponding
polynomial functor $P$ satisfies $P(X) = A_0 \times X$.

Using finality of $\M(A,B)$ we can define maps into streams and prove that they
have the expected computational behaviour.
For example, the $\zip$ function
\[
\zip : \stream(A) \times \stream(B) \to \stream(A \times B)
\]
can be obtained from the universal property applied to the coalgebra
\[
\begin{array}{l}
\theta : \stream(A) \times \stream(B) \to (A \times B) \times (\stream(A) \times \stream(B)) \\
\theta (xs, ys) := ((\head(xs), \head(ys)), (\tail(xs), \tail(ys))
\end{array}
\]
where $\head : \stream(X) \to X$ and $\tail : \stream(X) \to \stream(X)$ are the
two components of the final coalgebra $\out$.
The computational behaviour of $\zip$ is expressed by the fact that $\zip$ is a
coalgebra morphism
\[
\zip(xs, ys) = \cons((\head(xs), \head(ys)), (\zip (\tail(xs), \tail(ys)))),
\]
where $\cons = \out^{-1}$.
\end{example}

\section{Derivability of $\M$-types}\label{sec:derive-m}

In \Cref{sec:m-in-hott} we defined the type $\Final_S$ of final coalgebras of a signature $S$, and showed that this 
type is a proposition (\Cref{lem:is-prop-final}).
In this section, we construct an element of $\Final_S$, which, combined with \Cref{lem:is-prop-final}, proves the following theorem per the remark at the end of \Cref{sec:hott}:

\begin{theorem}\label{thm:is-contr-final}
 The type 
   \[\Final_S = \sm{(X,\rho)}\prd{(C,\gamma)}\iscontr\big(\Coalgmor{(C,\gamma)}{(X,\rho)}\big)\] 
 is contractible.
\end{theorem}

The construction of the final coalgebra is done in several steps, inspired by a construction of $\M$-types from $\W$-types 
in a type theory satisfying Axiom K by Altenkirch et al.~\cite{indexed_containers}.
Its carrier is defined as the limit of a \emph{chain}:

\begin{definition}
 A \fat{chain} is a pair $(X,\pi)$ of a family of types $X:\N \to \type$ and a family of functions $\pi_n : X_{n+1} \to X_n$.
 Here and below we write $X_n:=X(n)$ for the $n$th component of the family $X$.
\end{definition}

The (homotopy) limit of such a chain is given by the type of \enquote{compatible tuples}:
\begin{definition}\label{defn:limit}
 The \fat{limit} of the chain $(X,\pi)$ is given by the type
 \[ L := \sm{x:\prd{n:\N}X_n}\prd{n:\N}\pi_{n}x_{n+1} = x_{n} \enspace . \] 
The limit is equipped with projections $p_n : L \to X_n$, and $\beta_n : \comp{p_{n+1}}{\pi_n} = p_n$.  Sometimes, we simply write
\[ L = \lim X,\]
when the maps $\pi$ are clear.
\end{definition}

Note that this limit (we drop the adjective \enquote{homotopy}) is an instance of the general construction of homotopy limits 
by Avigad, Kapulkin, and Lumsdaine \cite{homotopy_limits}.

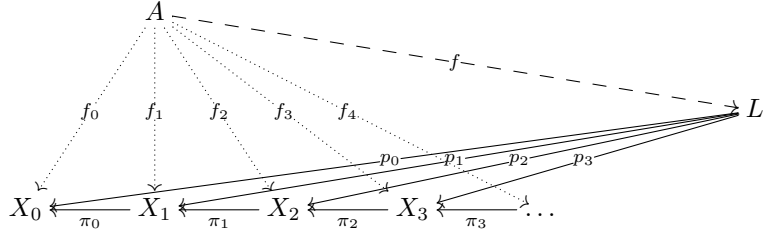
\begin{figure}[hbt]
\centering
  \[
   \begin{xy}
    \xymatrix@C=30pt{
             & A \ar@{.>}[ddl]|{f_0}\ar@{.>}[dd]|{f_1} \ar@{.>}[ddr]|{f_2} \ar@{.>}[ddrr]|{f_3} \ar@{.>}[ddrrr]|{f_4} \ar@{-->}[drrrrr]|{f}& \\
       & & & &  & &  L \ar[lllllld]|{p_0} \ar[llllld]|{p_1} \ar[lllld]|{p_2} \ar[dlll]|{p_3}&   \\
         X_0 & X_1 \ar[l]^{\pi_0} & X_2 \ar[l]^{\pi_1}& X_3 \ar[l]^{\pi_2} & \ldots \ar[l]^{\pi_3}  
     }
   \end{xy}
  \]
\caption{Universal property of $L$} \label{fig:universal}
\end{figure}

\begin{lemma}\label{lem:univ_prop_limit}
 The type $L$ satisfies the following universal property: for all types $A$, we have an equivalence of types between maps into $L$ and 
 \enquote{cones} over $X$:
 \[
      A \to L \enspace \simeq  \enspace \sm{f : \prd{n:\N} A \to X_n} \prd{n:\N} \comp{f_{n+1}}{\pi_n} = f_n  \enspace =: \Cone(A) \enspace . 
  \]
\end{lemma}
The equivalence, from left to right, maps a function $f : A \to L$ to
its projections $\comp{f}{p_n}$, as shown in \Cref{fig:universal}.

The next lemma is about tuples in \fat{co}chains, that is, tuples in chains with inverted arrows. Those tuples are determined by their first element:
\begin{lemma}\label{lem:first_elem_deter_cochain_tuple}
  Let $X : \N \to \type$ be a family of types, and $l : \prd{n:\N} X_n \to X_{n+1}$ a family of functions.
  Let 
  \[ Z := \sm{x:\prd{n:\N} X_n} \prd{n:\N} x_{n+1} = l_n(x_n) \enspace . \]
  Then the projection $Z \to X_0$ is an equivalence.
\end{lemma}

\begin{proof}
 Let $G$ be the functor defined by $G Y = 1 + Y$. Fix an element $z : Z$. 
 Then $z$ and $l$ together define a $G$-algebra structure on $X$, regarded as
a fibration over $\N$.  
Since $\N$ is the homotopy initial algebra of $G$, the
type $S z$ of algebra sections of $X$ is contractible.  But $Z$ is 
equivalent to

\[ \sm{z:X_0} \sm{x : \prd{n:\N}X_n} (x_0 = z) \times \left(\prd{n:\N} x_{n+1} = l_n(x_n)\right) \enspace , \]
which is exactly $\sm{z : X_0} S z \simeq X_0$.
\end{proof}

\begin{lemma}\label{lem:limit_of_shifted_chain}
Let $(X, \pi)$ be a chain, and let $(X', \pi')$ be the \emph{shifted} chain,
defined by $X'_n := X_{n+1}$ and $\pi'_n := \pi_{n+1}$.  Then the two chains have
equivalent limits.
\end{lemma}
\begin{proof}
 Let $L$ and $L'$ be the limits of $(X,\pi)$ and $(X',\pi')$, respectively. 
 We have
 \begin{align*}
 L' &\overset{\definemarker{1}}{\simeq} \sm{y:\prd{n:\N} X_{n+1}} \prd{n:\N} \pi_{n+1}y_{n+1} = y_n \\
    &\overset{\definemarker{2}}{\simeq} \sm{x_0 : X_0} \sm{y : \prd{n:\N} X_{n+1}} (\pi_0 y_0 = x_0) \times \left(\prd{n:\N}\pi_{n+1}y_{n+1} = y_n\right) \\
    &\overset{\definemarker{3}}{\simeq} \sm{x:\prd{n:\N} X_n} (\pi_0 x_0 = x_0) \times \left(\prd{n:\N}\pi_{n+1}x_{n+2} = x_{n+1} \right)\\
    &\overset{\definemarker{4}}{\simeq} \sm{x:\prd{n:\N} X_n}  \prd{n:\N}\pi_{n}x_{n+1} = x_{n} \\
    &\overset{\definemarker{5}}{\simeq} L \enspace ,
 \end{align*}
where \refermarker{1} and \refermarker{5} are by definition. Equivalence \refermarker{2} is given by multiplying with the contractible 
type $\sm{x_0:X_0} \pi_0 y_0 = x_0$ \cite[Lem.\ 3.11.8]{hottbook} and subsequent swapping of components in a direct product.
Equivalence \refermarker{3} is given by joining the first two components, and similarly in \refermarker{4} the last two components are joined.
\end{proof}

The next lemma says that  polynomial functors (see \Cref{def:poly_functor}) commute with limits of chains.
Let $(A,B)$ be a container with associated polynomial functor $P = P_{A,B}$.
Let $(X,\pi)$ be a chain with limit $(L,p)$. 
Define the chain $(PX,P\pi)$ with $PX_n := P(X_n)$ and likewise for $P\pi$, and let $L^P$ be its limit.
The family of maps $Pp_n : PL \to PX_n$ determines a function $\alpha : PL \to L^P$.
\begin{lemma}\label{lem:poly_continuous}
The function $\alpha$ is an isomorphism.
\end{lemma}
\begin{proof}
By \enquote{equational} reasoning we have
\begin{align*} L^P &\overset{\definemarker{6}}{\simeq} \sm{w : \prd{n:\N} \sm{a:A} B(a)\to X_n} \prd{n:\N} (P\pi_n) w_{n+1} = w_n  \\
                   &\overset{\definemarker{7}}{\simeq} \sm{a: \prd{n:\N}A} \sm{u:\prd{n:\N}B(a_n) \to X_n} \prd{n:\N} (a_{n+1},\comp{u_{n+1}}{\pi_n}) = (a_n,u_n) \\
                   &\overset{\definemarker{8}}{\simeq} \sm{a: \prd{n:\N}A} \sm{p:\prd{n:\N} a_{n+1}=a_n}\sm{u:\prd{n:\N} B(a_n) \to X_n} \prd{n:\N} \trans{(p_n)}{\comp{u_{n+1}}{\pi_n}} = u_n \\
                   &\overset{\definemarker{9}}{\simeq} \sm{a:A} \sm{u : \prd{n:\N} B(a) \to X_n} \prd{n:\N} \comp{u_{n+1}}{\pi_n} = u_n \\
                   &\overset{\definemarker{10}}{\simeq} \sm{a:A} B(a)\to L \\
                   &\overset{\definemarker{11}}{\simeq} PL
\end{align*}
where \refermarker{6} and \refermarker{11} are by definition, \refermarker{7} is by swapping $\Pi$ and $\Sigma$, \refermarker{8} by expanding equality of pairs,
\refermarker{9} by applying \Cref{lem:first_elem_deter_cochain_tuple} and \refermarker{10} by universal property of $L$.
Verifying that the composition of these isomorphisms is
$\alpha$ is straightforward.
\end{proof}

\begin{proof}[Proof of \Cref{thm:is-contr-final}]
We now construct a terminal coalgebra for a container $(A,B)$. Let $P = P_{A,B}$ the polynomial functor associated to the container.
By recursion on $\N$ we define the chain
\[
 \begin{xy}
  \xymatrix{
       1 & P1 \ar[l]_{!} & P^2 1 \ar[l]_{P!} & P^3 \ar[l]_{P^2 !} & \ldots \ar[l]_{P^3 !} 
   }
 \end{xy}
\]
which for brevity we call $(W,\pi)$, that is, $W_n:=P^n1$ and $\pi_n:=P^n!$.
 
Let $(L,p)$ be the limit of $(W, \pi)$.  If $L'$ is the limit of the shifted chain, we have a sequence of equivalences
\[  PL \overset{\definemarker{12}}{\simeq} L' \overset{\definemarker{13}}{\simeq} L \]
where \refermarker{12} is given by \Cref{lem:poly_continuous} and \refermarker{13} by \Cref{lem:limit_of_shifted_chain}.
 We denote this equivalence by $\in : P L \to L$, and its inverse by $\out : L \to P L$.

 It is worth noting that the construction of $L$ ``does not raise the universe
 level'', i.e., if $A$ and $B$ are contained in some universe $\mathcal{U}$,
 then $L$ is contained in $\mathcal{U}$ as well.  In other words, we only need
 one universe to carry out our construction of the final coalgebra.

 We will now show that $(L,\out)$ is a final $(A,B)$-coalgebra.  For this, let
$(C,\gamma)$ be any coalgebra, i.e., $\gamma : C \to PC$. The type of coalgebra
morphisms $\Coalgmor{(C,\gamma)}{(L,\out)}$ is given by
 \[ U:= \sm{f : C \to L} \comp{f}{\out} = \comp{\gamma}{Pf} \enspace .  \] We
need to show that $U$ is contractible.
We compute as follows (see below the math display for intermediate definitions):

  \begin{align*}
    U &\overset{\definemarker{14}}{\simeq} \sm{f : C \to L} \comp{f}{\out} = \comp{\gamma}{Pf} \\
      &\overset{\definemarker{15}}{\simeq} \sm{f : C \to L} \comp{f}{\out} = \step(f)\\
      &\overset{\definemarker{16}}{\simeq} \sm{f : C \to L} \comp{\comp{f}{\out}}{\in} = \comp{\step(f)}{\in}\\
      &\overset{\definemarker{17}}{\simeq} \sm{f : C \to L} f = \Psi(f)\\
      &\overset{\definemarker{18}}{\simeq} \sm{c : \Cone} e(c) = \Psi(e(c))\\
      &\overset{\definemarker{19}}{\simeq} \sm{c : \Cone} e(c) = e (\Phi(c))\\
      &\overset{\definemarker{20}}{\simeq} \sm{c : \Cone} c = \Phi(c)\\
      &\overset{\definemarker{21}}{\simeq} \sm{(u,q) : \Cone} \sm {p : u = \Phi_0 (u)} \trans{p}{q} = \Phi_1 (u)(q)\\
      &\overset{\definemarker{22}}{\simeq} \sm{u : \Cone_0}\sm{p : u = \Phi_0 u} \sm{q : \Cone_1 u} \trans{p}{q} = \Phi_1 (u)(q)\\
      &\overset{\definemarker{23}}{\simeq} \sm{t : 1} 1\\
      &\simeq 1
  \end{align*}
  where we use the following definitions:
  The function 
   $\step_Y : (C \to Y) \to (C \to PY)$ is defined as $\step_Y(f) := \comp{\gamma}{Pf}$ and
   $\Psi : (C \to L) \to (C \to L)$ is defined as $\Psi(f) := \comp{\step_L(f)}{\in}$.
  The map $\Phi : \Cone \to \Cone$ is the counterpart of $\Psi$ on the side of cones. We define
  $\Phi (u,g) = (\Phi_0 u, \Phi_1u(g)) : \Cone \to \Cone$ with 
  \begin{align*} (\Phi_0 u)_0 &:= x \mapsto tt : C \to 1 = W_0\\
                 (\Phi_0 u)_{n+1} &:= \step_{W_n}(u_n) : C \to W_{n+1} = PW_n
  \end{align*}
  and analogously for $\Phi_1$ on paths.
  By $e$ we denote the equivalence of \Cref{lem:univ_prop_limit} from right to left, and $\Cone = \sm{u:\Cone_0} \Cone_1(u)$ is short for $\Cone(C)$.
  The equivalence \refermarker{16} follows from $\in$ being an equivalence, and \refermarker{17} follows from $\in$ and $\out$ being inverse to each other.
  We pass from maps into $L$ to cones in \refermarker{18}, using the equivalence of \Cref{lem:univ_prop_limit}, while \refermarker{19} uses the commutativity of the following square:
\[
\xymatrix{
\Cone \ar[r]^-e \ar[d]_{\Phi} & (C \to L) \ar[d]^{\Psi}\\
\Cone \ar[r]^-e & (C \to L).
}
\]
  In \refermarker{21}, identity in a sigma type is reduced to identity of the components, and in \refermarker{22} the components are rearranged.
  Finally, step \refermarker{23} consists of two applications of \Cref{lem:first_elem_deter_cochain_tuple}.
  
  Altogether, this shows that for any coalgebra $(C,\gamma)$, the type of coalgebra morphisms $\Coalgmor{(C,\gamma)}{(L,\out)}$ is contractible.
  This concludes the construction of a final coalgebra associated to the signature $(A,B)$, and thus,
  combined with \Cref{lem:is-prop-final}, the proof of \Cref{thm:is-contr-final}. 
\end{proof}

From the construction of $L$ we get the following corollary about the homotopy level of $\M$-types:

\begin{lemma}
 The homotopy level of the (carrier of the) $\M$-type associated to the signature $(A,B)$ is bounded by that of the type of nodes $A$, that is,
 \[ \isofhlevel{n}{A} \to \isofhlevel{n}{\M(A,B)} \enspace . \]
\end{lemma}

\begin{example} We continue the example of streams of \Cref{example:streams}, with $A = A_0$ the type of nodes.
 In that case, the chain considered in the proof of \Cref{thm:is-contr-final} is given by $W_n=P^n(1) = A^n$, and the map $\pi_n : A^{n+1} \to A^n$
 chops of the $(n+1)$th element of any $(n+1)$-tuple. The limit $L$ is hence given by $A^\N = (\N \to A)$.
 The type of streams over $A$ has the same homotopy level as the type $A$ of nodes.
\end{example}

We conclude this section with a proof of the \emph{principle of coinduction}:

\begin{definition}[Bisimulation]

  Let $(C , \gamma)$ be a coalgebra for some signature $S$ with associated polynomial functor $P$ and let $\mathcal{R} : C \to C \to \mathcal{U}$ be a binary relation.
  Define 
    \[\overline{\mathcal{R}} := \sm{a : C} \sm{b : C}\mathcal{R}(a)(b)\] 
  along with two projections $\pi^{\overline{\mathcal{R}}}_1 (a,b,p) := a$ 
  and $\pi^{\overline{\mathcal{R}}}_2 (a,b,p):= b$. 

  An \emph{$S$-bisimulation} is a relation $\mathcal{R}$ together with a map
  $\alpha_{\mathcal{R}} : \overline{\mathcal{R}} \to P(\overline{\mathcal{R}})$
  such that both $\pi^{\overline{\mathcal{R}}}_1$ and
  $\pi^{\overline{\mathcal{R}}}_2$ are $P$-coalgebra morphisms:

\begin{center}
  \begin{tikzpicture}
    \node (C1) at (-3,0) {$C$};
    \node (R) at (0,0)  {$\overline{R}$};
    \node (C2) at (3,0)  {$C$};

    \node (FC1) at (-3,-2) {$P(C)$};
    \node (FR) at (0,-2) {$P(\overline{R})$};
    \node (FC2) at (3,-2) {$P(C)$};

    \draw[->] (R) -- node[above] {$\pi^{\overline{\mathcal{R}}}_1$} (C1);
    \draw[->] (R) -- node[above] {$\pi^{\overline{\mathcal{R}}}_2$} (C2);

    \draw[->] (R) -- node[fill=white] {$ \alpha_{\mathcal{R}}$} (FR);
    \draw[->] (C1) -- node[fill=white] {$\gamma$} (FC1);
    \draw[->] (C2) -- node[fill=white] {$\gamma$} (FC2);

    \draw[->] (FR) --node[below] {$P(\pi^{\overline{\mathcal{R}}}_1)$} (FC1);
    \draw[->] (FR) --node[below] {$P(\pi^{\overline{\mathcal{R}}}_2)$} (FC2);
  \end{tikzpicture}
\end{center}

\noindent
We say that a bisimulation is an \emph{equivalence bisimulation} when the
underlying relation is an equivalence relation.

\end{definition}

\begin{lemma}
  The identity relation $\cdot = \cdot$ over an $S$-coalgebra $C$ is an
  equivalence bisimulation. We write $\Delta_C$ for $\overline{\cdot = \cdot}$.
\end{lemma}

\begin{theorem}[Coinduction proof principle]
  Let $(L, \out)$ be the final coalgebra for $S$. For any bisimulation
$\overline{\mathcal{R}}$ over $L$, we have $\overline{\mathcal{R}} \subseteq
\Delta_L$. That is, for any $m,m' : L$,
 \[
  \mathcal{R}(m)(m') \to m = m'.
 \]
\end{theorem}
\begin{proof}
\
Since $(L,\out)$ is the final coalgebra, for any coalgebra $(C,\gamma)$ there exists a
unique coalgebra morphism $\mathsf{unfold}_C : C \to L$. It
follows that $\pi^{\overline{\mathcal{R}}}_1 =
\mathsf{unfold}_{\overline{\mathcal{R}}} =
\pi^{\overline{\mathcal{R}}}_2$. 
Finally, given $r : \mathcal{R}(m)(m')$, we obtain
$m = \pi^{\overline{\mathcal{R}}}_1 (m , m' , r) =
\pi^{\overline{\mathcal{R}}}_2 (m , m' , r) = m'$.
\end{proof}

\section{Indexed $\M$-types}\label{sec:indexed}

In this section, we briefly state the main definitions for \emph{indexed} $\M$-types. 
The difference to plain $\M$-types is that the type of nodes of an indexed $\M$-type is actually given by a family of types,
indexed by a type of sorts $I$.

\begin{definition}
 An \fat{indexed container} is given by a quadruple $(I,A,B,r)$ such that $I : \type$ is a type, $i : I\entails A(i) : \type$ is a family of types 
 dependent on $I$, $B$ is a family $i:I,a:A(i) \entails B_i(a) : \type$ and $r$ specifies the \enquote{sort} of the subtrees, i.e.,
  $r : \prd{i:I}\prd{a : A(i)} B_i(a) \to I$.
\end{definition}

\begin{definition}
 The polynomial functor $P$ associated to an indexed container $(I,A,B,r)$ is an endofunction on the type $I \to \type$:
  \[ (P X)(i) :=  \sm{a : A(i)} \prd{b : B_i(a)} X (r_{i,a}(b))  \enspace . \]
  The functorial action on morphisms is, analogously to \Cref{def:poly_functor}, given by postcomposition.
\end{definition}

Coalgebras for indexed containers, and their morphisms, are defined completely analogously to \Cref{def:coalgebra}.
Again, we prove that terminal coalgebras for indexed containers exist uniquely:

\begin{theorem}
  Let $(I,A,B,r)$ be an indexed container. 
  Then the associated indexed $\M$-type is uniquely specified and can be constructed in the type theory described in \Cref{sec:hott}.
\end{theorem}

An example of a coinductive type that needs indices to be expressed as an $\M$-type is a coinductive formulation of \emph{equivalence of types},
to our knowledge due to T.\ Altenkirch:

\begin{example}\label{ex:equivalence}
 Let $I:= \type \times \type$, and let $A : I \to \type$ be defined by $A(X,Y) := (X \to Y) \times (Y \to X)$.
 Define $B$ as $B(X,Y)(f,g) := X\times Y$ and $r(X,Y)(f,g)(x,y) := \left(f(x) = y\right) \times \left(x = g(y)\right)$.
 Then the associated $\M$-type is a family $\M : I\to \type$ and $M(A,B)$ is equivalent to $A \simeq B$. 
\end{example}

\section{Formalization}\label{sec:formal}
The proofs contained in this paper have been formalised in the proof assistant
\Agda in a self-contained development. 
The proof has been type-checked by version 2.4.2.2 of \Agda. 
The source code as well as HTML documentation can be found on
\url{https://github.com/HoTT/m-types/}. 
The \Agda source code is also part of this arXiv submission.

The formalised proofs deal with the indexed case (\Cref{sec:indexed}) directly,
but apart from that, they correspond closely to the informal proofs presented
here. In particular, they make heavy use of the ``equational reasoning''
technique to prove equivalences between types.

In fact, it is often the case that proving an equivalence between types $A$ and
$B$ ``directly'', i.e. by defining functions $A \to B$ and $B \to A$, and then
proving that they compose to identities in both directions, is unfeasibly hard,
due to the complexity of the terms involved.

However, in most cases, we can construct an equivalence between $A$ and $B$ by
composition of several simple equivalences.  Those simple building blocks range
from certain ad-hoc equivalences that make specific use of the features of the
two types involved, to very general and widely applicable ``rewriting rules'',
like the fact that we can swap a $\Sigma$-type with a $\Pi$-type (sometimes
called the \emph{constructive axiom of choice} \cite{hottbook}).

By assembling elementary equivalences, then, we automatically get both a
function $A \to B$ \emph{and} a proof that it is an equivalence.  However,
sometimes care is needed to ensure that the resulting function has the desired
computational properties.

An important consideration during the formalisation of the proof of 
\Cref{thm:is-contr-final} was keeping the size of the terms reasonably short.  For
example, in one early attempt, the innocent-looking term $\in$ was being
normalised into a term spanning more than 12000 lines.

The explosion in term size was clearly causing performance issues during
type-checking, which resulted in \Agda running out of memory while checking
apparently trivial proofs.

We solved this problem by moving certain definitions (like that of $\in$ itself)
into an \emph{abstract} block, thereby preventing \Agda from expanding it at all.
Of course, this means that we lost all the computational properties of certain
functions, so we had to abstract out their computational behaviour in the form
of propositional equalities, and manually use them in the proof of
\Cref{thm:is-contr-final}. This work-around is the source of most of the
complications in the formal proof.

\section{Conclusion and future work}

We have shown how to construct a class of coinductive types from the basic type
constructors and natural numbers in Homotopy Type Theory.  Our construction
follows a well-known pattern, that is known to work in type theory with identity
reflection rule.

We work in a univalent universe, but we make minimal use of univalence itself:
\Cref{lem:is-prop-final} is the only result that uses it directly, and elsewhere
univalence only appears indirectly through the use of functional extensionality.
We intend to make these dependencies more explicit in future developments of the
formalisation.

Finally, the coinductive types we construct do not satisfy the expected
computation rules \emph{judgmentally}, but only \emph{propositionally}.  This
fact would justify adding coinduction as a primitive rather than a derived
notion---provided that judgmental computation rules are validated by the
intended semantics. Those semantic questions are left for future work.

\bibliographystyle{plain}
\bibliography{literature}

\end{document}